\documentclass[12pt,a4paper]{article}
\usepackage{graphicx}
\usepackage[cp1251]{inputenc}
\usepackage{amssymb, amsfonts, amsmath}

\newtheorem{utv}{Theorem}
\newtheorem{slt}{Corollary}

\newenvironment{proof}[1][Proof]
{\textbf{#1.} }{\ \rule{0.5em}{0.5em}}

\begin{document}

\title{Yang--Mills gauge fields conserving the symmetry algebra of the Dirac equation in a homogeneous space}

\author{A.I. Breev\footnotemark[1] \footnotemark[2]\ \ and A.V. Shapovalov\footnotemark[1] \footnotemark[2]}

\date{}

\renewcommand{\thefootnote}{\fnsymbol{footnote}}
\footnotetext[1]{Tomsk State University, Lenin avenue 36, Tomsk 634050, Russia.}
\footnotetext[2]{Tomsk Polytechnic University, Lenin avenue 30, Tomsk 634050, Russia.}

\maketitle

\begin{abstract}
We consider the Dirac equation with an external Yang--Mills gauge field in a homogeneous space with an invariant metric. The Yang--Mills fields for which the motion group of the space serves as the  symmetry group for  the Dirac equation are found by comparison of the Dirac equation with an invariant matrix differential operator of the first order. General constructions are illustrated by the example of de Sitter space. The eigenfunctions  and the corresponding  eigenvalues for the Dirac equation are obtained in the space $\mathbb{R}^2\times \mathbb{S}^2$  by a noncommutative integration method.
\end{abstract}

\section*{Introduction}

Exact integration of relativistic wave equations for strong external fields is the topical problem in studying various effects in quantum field theory and cosmology where the standard $S$-matrix method does not work \cite{Birrel}.

The main technique for exact integration of the equations is based on the classical method of separation of variables (SoV) \cite{Kalnins,Kalnins2,Miller}. There are a large number of works dealing with classification of external fields admitting separation of variables in relativistic quantum equations (see, e.g.,  \cite{Bagrov} and references therein). In this connection, the  integration of relativistic wave equations with external fields by means of methods other than the SoV method can provide new possibilities to study the relativistic quantum wave  equations for classical and quantum fields and their interactions.

A new method of exact integration of linear partial differential equations was proposed \cite{SpSh1} and applied to quantum equations \cite{ShSO3,ShBr11,ShBr14}. This method differs from the classical SoV method and uses non-commutative algebras to describe the symmetry operators of the equation under consideration.

In this work, we consider the non-commutative symmetries of the Dirac equation with the potential of an external gauge Yang-Mills field on some homogeneous space.

The generators of a Lie transformation  group acting on the homogeneous space form a non-commutative symmetry algebra for the Dirac equation. So a Lie transformation group of the homogeneous space will be the symmetry group of the Dirac equation. In general,  the symmetry of the Dirac equation breaks down in the presence of an external gauge field. Our aim is to answer the question: what are the  Yang-Mills gauge fields which do not destroy the symmetry group of the Dirac equation in the homogeneous space? Note that this problem was solved for the Klein-Gordon equation \cite{KyrnScalar}. The symmetry of the Dirac equation in the presence of external gauge fields was also studied by the authors of Refs. \cite{BagrShD} and \cite{BagrQG}. We adopt natural units $\hbar = c = G = 1$, unless stated otherwise.

The paper is organized as follows: In Section \ref{invariant-mertic} we briefly describe the
necessary concepts and notations related to homogeneous spaces \cite{DiffG, Kobyashi}.

The construction of an invariant differential operator with matrix coefficients on a homogeneous space is introduced following Refs. \cite{KyrnScalar, BarProlong},in Section \ref{invariant-matrix-operator}.

In Section \ref{dirac-eq-in homogen-space} we present the Dirac equation on a homogeneous space with an invariant metric tensor in terms of a first-order invariant matrix operator. The spinor connection and the symmetry operators of the Dirac equation are shown to define the isotropy representation of a spinor space. The generators of the spinor representation are found explicitly.

 Next, in Section \ref{dirac-eq-gauge-field},  we find the gauge potentials of an external Yang--Mills field that do not change the symmetry algebra of the Dirac equation. For these potentials, the Dirac equation is presented as a system of equations on a motion group.

An illustration of the general results obtained is given in Section \ref{dirac-eq-r2-s2} by the example of the homogeneous space $\mathbb{R}^2\otimes \mathbb{S}^2$ with an invariant metric.  The homogeneous space is shown to admit the external magnetic field preserving the symmetry algebra of the Dirac equation. The spectrum of the Dirac equation and the corresponding eigenstates are found by a noncommutative integration method.

In the final Section \ref{de-sitter}, we find the Yang--Mills gauge fields preserving the symmetry of the Dirac equation in de Sitter space.

In Section \ref{conclusion} we  give our conclusion remarks.

\section{An invariant metric on a homogeneous space}
\label{invariant-mertic}

Here we provide some basic concepts and notations relevant to the theory of homogeneous spaces.

Let $G$ be a  simply connected real Lie group with a Lie algebra  $\mathfrak{g}$ and let $M$  be a homogeneous space  with right action of the group $G$. For any  $x\in M$ there exists an isotropy subgroup $H_x  \in G$. Denote by $H$ a stationary subgroup of a point $x_0\in M$, and  let $\mathfrak{h}$  be the Lie algebra of $H$. The homogeneous space $M $ is diffeomorphic to the manifold $G/H$ of right cosets $Hg$, where $H$ is the isotropy subgroup. The group of transformations  $G$ can be regarded as a principal bundle $(G,\pi,M,H)$ with the structure group $H$, the base $M$,  and the projection map $\pi: G\rightarrow M$. The Lie algebra $\mathfrak{g}$ is decomposed into a direct sum of subspaces $\mathfrak{g} = \mathfrak {h} \oplus \mathfrak {m}$, where $\mathfrak {m}\simeq T_{x_0} M $ is complement to $ \mathfrak {h}$.

The coordinates of an arbitrary element of $g\in G$ can be written as $g=hs(x)$, $h\in H$, i.e., $g^A=(x^a, h^{\alpha})$, where $A,B,C,\ldots = 1,\dots,\operatorname{dim}\mathfrak{g}; a,b,c,\ldots=1, \dots  ,\operatorname*{dim} M; \alpha,\beta,\gamma,\ldots = 1,\dots,\operatorname{dim}\mathfrak{h}$, and $s: M\rightarrow G $ is a local and smooth section of $G$.

Let us introduce  an invariant metric on the homogeneous space $M$. Suppose that $\textbf{G}$ is a non-degenerate  $Ad_H$-\textit{invariant} quadratic form on a subspace $\mathfrak{m}\subset\mathfrak{g}$, \begin{equation}\label{conds_AdH}
	\textbf{G}(\overline{[X,Y]},\overline{Z})+\textbf{G}(\overline{Y},\overline{[X,Z]})=0,
   \quad X\in \mathfrak{h}, \quad Y,Z\in \mathfrak{g},
\end{equation}
where the bar denotes the projection map of the Lie algebra  $\mathfrak{g} $ onto the subspace $\mathfrak {m}$. The quadratic form $\textbf{G}$ defines an invariant inner product $Ad_H$ on the tangent space $T_{x_0}M \simeq \mathfrak{m}$. By the action of a Lie group $G$ with right-hand shifts $R_g$ ($g\in G$) in the homogeneous space $M$, we define the inner product throughout the space  $M$ as
\begin{equation}\label{inv_metric}
   g_M(\tau , \tau')(x) = \mathbf{G}((R_{g^{-1}})_{*}\tau,(R_{g^{-1}})_{*}\tau'),
   \quad \tau,\tau' \in T_x M,\quad x=\pi (g).
\end{equation}
The $Ad_H$-invariance (\ref{conds_AdH}) is the necessary and sufficient condition
for the inner product (\ref{inv_metric}) to be invariant with respect to the Lie group $G$ action on the homogeneous space $M$. The inner product (\ref{inv_metric}) defines an  \textit{invariant} metric $g_M$ on the homogeneous space $M$ \cite{Kobyashi}. From (\ref{inv_metric}) we can write down the covariant components of the  metric tensor in local coordinates as $(i,j,k,l,\ldots=1,\dots,\operatorname{dim}M)$
\begin{equation}\label{gij_loc}
  g_{ij}(x) =
  G_{a b}\sigma^a_i(x,e_H)\sigma^b_j(x,e_H),
  \quad G_{a b} \equiv \textbf{G}(e_a, e_b),\quad a,b = 1,\dots,\dim M.
\end{equation}
Here $e_a$  are the fixed basis vectors of the space $\mathfrak{m}$, $\sigma^b(g)\equiv - (R_g)^{*}e^b$ is the basis of right-invariant 1-forms, and  $e^b$  are the basis vectors in the dual space $\mathfrak{m}^{*}$: $\langle e_a,e^b\rangle = \delta_a^b$, and $e_H$ is the identity element of $H$. The contravariant components of the metric tensor (\ref{gij_loc}) can be represented as
\begin{equation}
	g^{ij}(x) = G^{a b}\eta_a^i(x,e_H)\eta_b^j(x,e_H),
	\quad G^{a b} = (G_{a b})^{-1},\quad \eta_a^i = (\sigma^a_i)^{-1},
	\nonumber
\end{equation}
where $\eta_a(g) = -(R_g)_* e_a$  are right-invariant vector fields on the Lie group $G$.
On the basis of the algebra $\mathfrak{g}$, the $Ad_H$-invariance condition  (\ref{conds_AdH}) takes the form
\begin{equation}
	G_{ab}C^a_{c \alpha} + G_{ac} C^a_{b \alpha } = 0,
	\label{AdH_comp}
\end{equation}
where $C^A_{AB} = [e_A,e_B]^C$ are the structure constants of $\mathfrak{g}$.

Let us remark that any non-degenerate  symmetric 2-form $G_{ab}$ satisfying condition (\ref{AdH_comp}) defines an invariant metric  $g_M$ on the homogeneous space $M$.

The metric tensor (\ref{gij_loc}) defines the Christoffel symbols of the Levi-Civita connection
as \cite{ShBr11,Kobyashi}
   \begin{align}\label{gamma_P_IJK}
	   \Gamma_{jk}^i(x) = & {\Gamma}_{b c}^a \sigma^b_j(x,e_H)\sigma^c_k(x,e_H)\eta_a^i(x,e_H) - \\ \nonumber
      & -\sigma^b_j(x,e_H)\eta_{b,k}^i(x,e_H) - C_{b\alpha}^a\sigma^b_j(x,e_H)\sigma_k^{\alpha}(x,e_H)\eta_a^i(x,e_H).
   \end{align}
 The coefficients $\Gamma_{bc}^a$ are determined by the components $G^{ab}$ of the quadratic form $\mathbf {G}$ and by the structure constants of the Lie algebra:
  \begin{equation}\label{gamma_P}
     \Gamma_{b c}^a = -\frac{1}{2}C_{bc}^a - \frac{1}{2}G^{ad}
     \left[G_{e c}C_{b d}^e + G_{e b} C_{c d}^e\right].
  \end{equation}
Thus, the Levi-Civita connection is defined by the algebraic properties of a homogeneous space with an invariant metric.

\section{Invariant matrix differential operator of the first order}
\label{invariant-matrix-operator}

In this section, we consider algebraic conditions for a first-order linear differential operator with matrix coefficients invariant on the homogeneous space $M$. We follow Ref. \cite{KyrnScalar}, which presents a study of a more general invariant linear matrix differential operator of the second order was studied.

Denote by $C^{\infty}(M,V)$ and $C^{\infty}(G,V)$ two spaces of functions taking values on a linear space $V$ and defined on a homogeneous space $M$ and on a transformation group $G$, respectively.

The functions on the homogeneous space $M$ can be considered as functions defined on the Lie group $G$ and  invariant along the fibers $H$ of the bundle $G$ \cite{Kobyashi}. In our case, if functions take values on the vector space $V$, the space $C^{\infty}(M,V)$ is isomorphic to the function space
\begin{equation}\nonumber
	\hat{\mathcal{F}} = \{ \varphi \in C^{\infty}(G,V) \mid	\varphi(h g) = U(h)\varphi(g),\quad h\in H  \},
	\label{def_hatF}
\end{equation}
where $U(h)$ is an exact representation of the isotropy group $H$ on $V$. For an arbitrary function $\varphi \in\hat{\mathcal{F}}$, we have
\begin{equation}
	\varphi(g)=\varphi(hs(x))=U(h)\varphi(s(x)), \quad g = (x,h).	
	\label{condF}
\end{equation}
Then we can identify $\varphi(s(x))$ with a function $\varphi\in C^{\infty}(M,V)$. Formula (\ref{condF}) gives an explicit form of the isomorphism $\hat{\mathcal{F}}\simeq C^{\infty}(M,V)$. Differentiating (\ref{condF}) with respect to $h^\alpha$ and assuming $h=e_H$, we obtain:
\begin{equation}
	\left( \eta_{\alpha} + \Lambda_{\alpha} \right)\varphi(g) = 0,\quad
	\Lambda_{\alpha} = \frac{\partial U(h)}{\partial h^\alpha}|_{h = e_H}.
	\label{inf_hatF}
\end{equation}
Here $\Lambda_{\alpha}$ are generators of the group $H$ in the linear space $V$. Formula (\ref{inf_hatF}) is the  infinitesimal consequence of (\ref{condF}). The isotropy subgroup $H$ is assumed to be  connected. Then conditions (\ref{condF}) and (\ref{inf_hatF}) are equivalent.

From (\ref{inf_hatF}) we can see that a linear differential operator $R$ leaves invariant the function space $\hat{ \mathcal{F}}$ if
\begin{equation}
	(\eta_\alpha + \Lambda_\alpha)R\varphi(g) = [\eta_\alpha + \Lambda_\alpha,R]\varphi(g) = 0,\quad \varphi\in\hat{\mathcal{F}}.
	\label{ravFF}
\end{equation}
%

Denote by $L(\hat{\mathcal{F}})$ a space of linear differential operators on $C^{\infty}(G,V)$ satisfying the condition
\begin{equation}
	[ \eta_\alpha + \Lambda_\alpha , R ]|_{\hat{\mathcal{F}}} = 0, \quad
	\alpha = 1,\dots,\dim\mathfrak{h}.
	\label{condLF}
\end{equation}
Then, given relation (\ref{condF}), the action of $R\in L(\hat{\mathcal{F}})$ on a function $\varphi (g)$ from the space $\hat{\mathcal{F}}$ can be written as
\begin{equation}
	R \varphi(g) = U(h)\left( U^{-1}(h) R U(h) \right)\varphi(s(x)).
	\label{actR}
\end{equation}
Multiplying equation (\ref{ravFF}) by $U^{-1}(h)$ and taking into account that $\eta_\alpha U(h)=-\Lambda_\alpha U(h)$, we obtain
\begin{gather}\nonumber
	U^{-1}(h)[\eta_\alpha + \Lambda_\alpha, R] U(h)\varphi(s(x)) =
	[\eta_\alpha, U^{-1}(h)R U(h)]\varphi(s(x)) = \\ \nonumber
	\eta_\alpha \left( U^{-1}(h) R U(h) \varphi(s(x))\right) = 0.
\end{gather}
Hence, the operator $U^{-1}(h)RU(h)$ is independent of $h$ and (\ref{actR}) can be presented as
\begin{equation}
	R \varphi(g) = U(h) R_M\varphi(s(x)),\quad
	R_M \equiv \left( U^{-1}(h) R U(h) \right)|_{h = e_H} =
	R U(h)|_{h = e_H}.
	\label{actRM}
\end{equation}
Thus, for any operator $R$ from  $L(\hat{\mathcal{F}})$ there exists an operator $R_M$ in a homogeneous space $M$ which acts on the acts on the functions of the space $C^{\infty}(M, V)$.
We call $R_M$ \textit{the projection operator}, $R_M=\pi_{*}R$.
For example, for a linear differential operator of the first order
$$
	R_1 = B^a(x,h)\partial_{x^a} + B^\alpha(x,h)\partial_{h^\alpha} + B(x,h),
$$
the projection is:
\begin{equation}
	R^{(1)}_M = \pi_{*}R_1 =  B^a(x,e_H)\partial_{x^a}+ B^\alpha(x,e_H)\Lambda_\alpha + B(x,e_H).		
	\label{piR1}
\end{equation}
On the other hand, any linear differential operator $R_M$ defined on $C^{\infty}(M, V)$ corresponds to an operator $R =U(h)R_M U^{-1}(h) \in L(\hat{\mathcal{F}})$.
Thus we have the isomorphism $L(\hat{\mathcal{F}})\simeq L(C^{\infty}(M,V))$ given by  (\ref{actRM}).

Let $\xi_X (g)$ be a left-invariant vector field on a Lie group $G$, $X\in \mathfrak{g}$. Since the left-invariant vector fields commute with the right-invariant ones, condition (\ref{condLF}) is fulfilled. Using (\ref{piR1}), we find the corresponding operator on the homogeneous space in the form:
\begin{equation}
		X' = \pi_{*}\xi_X = \xi_X^a(x)\partial_{x^a} + \xi_X^\alpha(x,e_H)\Lambda_\alpha,\quad
		X \in \mathfrak{g}.
	\label{xiX}
\end{equation}
It is easy to verify the following commutation relations:
\begin{gather}
\nonumber
	[X', Y'] = [U^{-1}(h)\xi_X U(h),U^{-1}\xi_Y U(h)]|_{h = e_H} =
	[U^{-1}(h) [\xi_X,\xi_Y] U(h)]|_{h = e_H} = \\ \nonumber
	[U^{-1}(h) \xi_{[X,Y]} U(h)]|_{h = e_H} = [X,Y]'
\end{gather}
for all $X,Y$, of the form (\ref{xiX}), $X,Y\in\mathfrak{g}$.
Consequently, the operators $X'$ corresponding to the left-invariant vector fields
$\xi_X$ are the \textit{generators} of a transformation group acting on  $C^{\infty}(M, V)$.

An operator $R_M\in L(C^{\infty}(M, V))$ is invariant under the action of a Lie group of transformations, if $R_M$ commutes with $X'$:
\begin{equation}\label{commRX}
	[R_M, X'] = [U^{-1}(h) R U(h), U^{-1}(h)\xi_X U(h)]|_{h = e_H} =
	U^{-1}(h)[R,\xi_X]U(h)|_{h = e_H} = 0.
\end{equation}
From (\ref{commRX}) it follows that the operator $R_M$ is invariant with respect to the transformation group  if and only if the corresponding operator $R\in L(\hat{\mathcal{F}})$ commutes with the left-invariant vector fields:
\begin{equation}
	[R,\xi_X] = 0,\quad X\in\mathfrak{g}.
	\label{commXi}
\end{equation}
Suppose that $R^{(1)}_M \in C^{\infty}(M, V)$  is a first-order linear differential operator invariant with respect to the group action. By virue of (\ref{commXi}), this operator corresponds to a first-order polynomial of the right-invariant vector fields:
\begin{equation}\nonumber
	R_{(1)} = B^a \eta_a(x,h) + B^\alpha \eta_\alpha(h) + \tilde B.
\end{equation}
The projection map of $B^\alpha\eta_\alpha(h)$ is a constant $B^\alpha\Lambda_\alpha$, which can be removed from the operator $R^{(1)}_M$ by changing the variable $B=\tilde B+B^\alpha \Lambda_\alpha$.
Therefore,  we can put $B^\alpha=0$ without loss of generality.
Substituting the operator $R_{(1)}$ in condition (\ref{condLF}), we get
\begin{gather}\nonumber
	[\eta_\alpha + \Lambda_\alpha, R_{(1)}]|_{\hat{\mathcal{F}}} =
	\left([b^a,\Lambda_\alpha]\eta_a + b^a[\eta_a,\eta_\alpha] + [B,\Lambda_\alpha]\right)_{\hat{\mathcal{F}}} =\\ \nonumber
	\left( [B^a,\Lambda_\alpha] + B^b C_{b\alpha}^a \right)\eta_a|_{\hat{\mathcal{F}}} +
	[B,\Lambda_\alpha] - B^a C^\beta_{a\alpha}\Lambda_\beta = 0.
\end{gather}
Also, we have the following  system of algebraic equations for the coefficients $B^a$ and $B$:
\begin{align}
	\label{sysBa}
	& [B^a,\Lambda_\alpha] + B^b C_{b\alpha}^a = 0,\\
	\label{sysB}
	& [B,\Lambda_\alpha] - B^a C^\beta_{a\alpha}\Lambda_\beta = 0.	
\end{align}
Under  the conditions (\ref{sysBa})--(\ref{sysB}), the  projection map of $R_{(1)}$ on the homogeneous space yields the desired form of the first-order invariant linear differential operator:
\begin{equation}
	R^{(1)}_M = \pi_{*} R_{(1)} =
	B^a\eta^i_a(x,e_H)\partial_{x^i} + B^a\eta^\alpha_a(x,e_H)\Lambda_\alpha + B.
	\label{R1form}
\end{equation}

Thus, any linear first-order differential operator  acting on functions from  $C^{\infty}(M,V)$ and invariant with respect to the action of the transformation group has the form (\ref{R1form}), where the matrix coefficients $B^a$ and $B$ satisfy the algebraic system of equations (\ref{sysBa})--(\ref{sysB}).

\section{The Dirac equation on a homogeneous space}
\label{dirac-eq-in homogen-space}

In this section, we consider the Dirac equation on a four-dimensional homogeneous space $M$.
We assume that in the four-dimensional homogeneous space $M$, an invariant metric $g_M$ of signature $(+, -, -, -)$ and a Levi-Civita connection are given. Denote by $V_\Psi$  a  space of spinor fields on $M$.

Let us write down the Dirac equation on the space $M$ as an equation on a four-dimensional Lorentzian manifold $M$ \cite{BagrQG} as follows:
\begin{equation}
	\left(i\gamma^k(x)[ \nabla_k + \Gamma_k(x)  ] - m\right)\psi(x) = 0.
	\label{diracM}
\end{equation}
Here $\nabla_k$ is the covariant derivative corresponding to the Levi-Civita connection on $M$ and $m$ is the mass of the field $\psi \in C^{\infty}(M,V_\Psi)$. The Dirac gamma matrices, $\gamma^k(x)$,  satisfy the condition
\begin{equation}
	\{\gamma_i(x),\gamma_j(x)\}=2 g_{ij}(x) E_4,	
	\label{sys_gamma}
\end{equation}
where $E_4$ denotes an identity matrix. The spinor connection $\Gamma_k(x)$ satisfies the conditions $ [\nabla_k + \Gamma_k(x),\gamma_j(x)]=0,\ \operatorname*{Tr} \Gamma_k(x) =0$ and can be presented in explicit form as \cite{BagrQG}:
$$
  \Gamma_k(x) = -1/4 (\nabla_k\gamma_{j}(x))\gamma^j(x).
$$
We will seek a solution to (\ref{sys_gamma}) in terms of a tetrad decomposition:
\begin{equation}
	\gamma^k(x) = \hat\gamma^a \eta^k_a(x,e_H),\quad
	\hat\gamma^b =\gamma^k(x)\sigma^b_k(x,e_H).
	\label{my_sol_gamma}
\end{equation}
The constant matrices $\hat\gamma^a$ are the tetrad components of $\gamma^k(x)$ and satisfy the system of algebraic equations
\begin{equation}
	\{ \hat\gamma^a, \hat\gamma^b \} = 2 G^{ab} E_4.
	\label{sys_gamma2}
\end{equation}
It follows from (\ref{gij_loc}) that the gamma matrices with subscripts are of the form:
\begin{equation}
	\gamma_i(x) = g_{ij}(x)\gamma^j(x) = \hat\gamma_a\sigma^a_i(x,e_H),\quad
	\hat\gamma_a \equiv G_{ab}\hat\gamma^b.
	\label{gamma_down}
\end{equation}
The spinor connection is defined by the following theorem.
\begin{utv}
\label{prop1}
	Let $\Gamma(x) = \gamma^k(x)\Gamma_k(x)$ be the spinor connection on a four-dimensional homogeneous space $M$ with an invariant metric  $g_M$. Then we have
	\begin{equation}
		\Gamma(x) = \hat\gamma^a\left( \Gamma_a + \eta^\alpha_a(x,e_H)\Lambda^s_\alpha \right),\quad	\Gamma_a = -\frac{1}{4}\Gamma^d_{b a}\hat\gamma^b\hat\gamma_d,\quad \Lambda^s_\alpha = -\frac{1}{8}G_{a c} C^a_{\alpha b}[\hat\gamma^b,\hat\gamma^c].
		\label{sp_gamma_x}
	\end{equation}
\end{utv}
\begin{proof}
Let us write $\Gamma (x)$ so that the covariant derivative $\nabla_j \gamma_k$ be expressed in terms of the Christoffel symbols $\Gamma^l_{kj}(x)$:
\begin{equation}
		\Gamma(x) = \frac{1}{4}\gamma^j(x)\gamma^k(x)\left( \partial_{x^j}\gamma_k(x) - \Gamma^l_{kj}(x)\gamma_l(x) \right).
		\label{connection-1}
	\end{equation}
Substituting  the Christoffel symbols  (\ref{gamma_P_IJK}) and expressions (\ref{my_sol_gamma}) and (\ref{gamma_down}) to (\ref{connection-1}), we obtain
\begin{equation}\nonumber
		\Gamma(x) = \hat\gamma^a\Gamma_a + \frac{1}{4} C^d_{b\alpha}\hat\gamma^a\hat\gamma^b\hat\gamma_d \sigma^\alpha_j(x,e_H)\eta^j_a(x,e_H).
		\label{}
	\end{equation}
Using the property (\ref{AdH_comp}) of the invariant metric, we reduce the expression $C^d_{b\alpha}\hat\gamma^b\hat\gamma_d$ to
\begin{equation}\nonumber
		C^d_{b\alpha}\hat\gamma^b\hat\gamma_d =
		C^d_{b\alpha}G_{d c} \hat\gamma^b\hat\gamma^c =
		-C^d_{c\alpha} G_{d b} \hat\gamma^b\hat\gamma^c =
		 \frac{1}{2}C^d_{b\alpha} G_{d c }[\hat\gamma^b,\hat\gamma^c] = 4\Lambda^s_\alpha.
		\label{}
	\end{equation}
From the chain of equalities
\begin{gather}\nonumber
		\sigma^\alpha_j(x,e_H)\eta^j_a(x,e_H) =
		\sigma^\alpha_A(x,e_H)\eta^A_a(x,e_H) - \sigma^\alpha_\beta(e_H)\eta^\beta_a(x,e_H) =\\ \nonumber
		\delta^\alpha_a - (-\delta^\alpha_\beta)\eta^\beta_a(x,e_H) = \eta^\alpha_a(x,e_H),
	\end{gather}
we easily obtain the required expression (\ref{sp_gamma_x}) for the spinor connection.
\end{proof}

Thus we can write down the Dirac equation (\ref{diracM}) on the homogeneous space $M$ with the invariant metric $g_M$ and the gamma matrices of the form (\ref{my_sol_gamma}):
\begin{gather}\nonumber
	\mathcal{D}^0_M \psi = m \psi,\quad
	\mathcal{D}^0_M = i \hat\gamma^a \left[ \eta^j_a(x,e_H)\partial_{x^j} + \Gamma_a  + \eta^\alpha_a(x,e_H)\Lambda^s_\alpha \right].
\end{gather}
 The set of matrices $\Lambda^s_\alpha$ determines the  spinor representation of the group $H$ in the space $V_\Psi$.
\begin{utv}
\label{prop2}
The matrices  $\Lambda^s_\alpha$ are the generators of  the group $H$ representation on the space  $V_\Psi$.
\end{utv}
\begin{proof}
We clime that the matrices $\Lambda^s_\alpha$ satisfy the commutation relations
	\begin{equation}
		[\Lambda^s_\alpha,\Lambda^s_\beta] = C^\gamma_{\alpha\beta}\Lambda^s_\gamma.
		\label{genH4}
	\end{equation}
	Indeed, the commutator of two matrices $\Lambda^s_\alpha$ and $\Lambda^s_\beta$ can be written as:
\begin{gather}\label{comm3}
		[\Lambda^s_\alpha,\Lambda^s_\beta] = -\frac{1}{4} C^d_{\beta b}  [\Lambda^s_\alpha,\hat\gamma^b\hat\gamma_d] = -\frac{1}{4}C^d_{\beta b}\left( [\Lambda^s_\alpha,\hat\gamma^b]\hat\gamma_d + \hat\gamma^b [\Lambda^s_\alpha,\hat\gamma_d] \right).
	\end{gather}
	Using (\ref{AdH_comp}), (\ref{sys_gamma2}), and (\ref{gamma_down}), we find the commutator of $\Lambda_\alpha$
with the gamma matrices  $\hat\gamma^a$:
\begin{gather}\label{comm1}
		[\Lambda^s_\alpha, \hat\gamma^a] = \frac{1}{4}C^d_{d\alpha} [\hat\gamma^b\hat\gamma_d,\hat\gamma^a] =
		\frac{1}{2} C^d_{b\alpha}\left( \delta^a_d\hat\gamma^b - G^{ab}\hat\gamma_d \right) = \frac{1}{2}\left( G^{bd}C^a_{b\alpha} - G^{ab}C^d_{b\alpha} \right)\hat\gamma_d = C^a_{b\alpha}\hat\gamma^b.		
	\end{gather}
In the same way, we obtain for the gamma matrices with subscript indices:
\begin{equation}
		[\Lambda^s_\alpha,\hat\gamma_a] = C^b_{\alpha a}\hat\gamma_b.
		\label{comm2}
	\end{equation}
Substitution of (\ref{comm1})--(\ref{comm2}) in  (\ref{comm3}) yields
	\begin{equation}
		[\Lambda^s_\alpha, \Lambda^s_\beta] =\frac{1}{4}\left( C^d_{\beta e} C^e_{\beta b} - C^e_{\beta b}C^d_{\alpha e} \right)\hat\gamma^b\hat\gamma_d.
		\label{commLL}
	\end{equation}
	The expression  enclosed in the parentheses takes the form
	\begin{equation}
		C^d_{\beta e} C^e_{\beta b} - C^e_{\beta b}C^d_{\alpha e} =
		\left[C^A_{\alpha b} C^d_{\beta A} +  C^A_{b \beta}C^d_{\alpha A} + C^A_{\beta\alpha}C^d_{b A}\right] + C^\gamma_{\alpha\beta}C^d_{b \gamma}.
		\label{CC_CC}
	\end{equation}
Applying the Jacobi identity for structure constants to the expression in parentheses, we see that it vanishes.
Substituting (\ref{CC_CC}) in (\ref{commLL}), we obtain (\ref{genH4}).
\end{proof}

Let us associate the Dirac operator $\mathcal{D}^0_M$ with an operator $\mathcal{D}^0_G$ using a projection map $\pi_*$ similar to (\ref{R1form}).
\begin{utv}
\label{prop3}
	The Dirac operator $\mathcal{D}_M^0$ in the homogeneous space  $M$ with the invariant metric $g_M$ can be presented as:
\begin{equation}
	\mathcal{D}_M^0 = \pi_{*}\mathcal{D}_G^0,\quad
	\mathcal{D}_G^0 \equiv i\hat\gamma^a[\eta_a(g) + \Gamma_a] \in L(\hat{\mathcal{F}}_\Psi)
	\label{dirac02}
\end{equation}
\end{utv}
\begin{proof}
Comparing the Dirac operator $\mathcal{D}_M^0$ with the first-order invariant matrix differential operator (\ref{R1form}) on the  homogeneous space $M$, we obtain
\begin{equation}
		B^a = i\hat\gamma^a,\quad B = i \hat\gamma^a\Gamma_a.
		\label{coefBB}
	\end{equation}
The Dirac operator $\mathcal{D}_M^0$ in   (\ref{dirac02}) is defined if the coefficients  $B^a$  and  $B$ of the form (\ref{coefBB}) satisfy equations (\ref{sysBa})--(\ref{sysB}). From (\ref{comm1}) it follows that the commutator of $\Lambda^s_\alpha$ and  $\hat\gamma^a$ satisfies the first condition in (\ref{sysBa}). In this case, condition  (\ref{sysB}) is reduced to the expression
	\begin{equation}\label{condLs2}
	[\Gamma,\Lambda_\alpha] = C^\beta_{a\alpha}\hat\gamma^a\Lambda_\beta,\quad
    \Gamma = \hat\gamma^a\Gamma_a.
	\end{equation}	
The commutator  of $\Gamma$ and  $\Lambda^s_\alpha$ can be presented in terms of the commutator $[\Lambda^s_\alpha,\Gamma_a]$:
	\begin{equation}
		[\Lambda^s_\alpha,\Gamma] = [\Lambda^s_\alpha,\hat\gamma^a]\Gamma_a +
		\hat\gamma^a [\Lambda^s_\alpha,\Gamma_a].
		\label{commLG}
	\end{equation}	
Using (\ref{sp_gamma_x}), in view of properties  (\ref{comm1})--(\ref{comm2}), we obtain:
	\begin{equation}
		[\Lambda^s_\alpha,\Gamma_a] = -\frac{1}{4}\Gamma^d_{ba}\left( [\Lambda^s_\alpha,\hat\gamma^b]\hat\gamma_d +\hat\gamma^b [\Lambda^s_\alpha,\hat\gamma_d] \right) = -\frac{1}{4} \Gamma^d_{ba}\left( C^c_{\alpha d}\hat\gamma^b\hat\gamma_c - C^b_{\alpha c} \hat\gamma^c\hat\gamma_d \right).
		\label{commLGa}
	\end{equation}
	Substituting (\ref{commLGa}) in (\ref{commLG}), we have
	\begin{equation}
		[\Lambda^s_\alpha,\Gamma] = \frac{1}{4}\left( C^e_{\alpha a}\Gamma^c_{be} - C^c_{\alpha c} \Gamma^e_{ba} + C^e_{\alpha b}\Gamma^c_{ea}  \right)\hat\gamma^a\hat\gamma^b\hat\gamma_c.
		\label{commLG3}
	\end{equation}	
From (\ref{gamma_P}) and the Jacobi identity for the structure constants of the Lie algebra $\mathfrak{g}$, it follows that
\begin{equation}
		C^c_{\alpha e}\Gamma^e_{ba} = C^e_{\alpha a}\Gamma^c_{be} + C^e_{\alpha b}\Gamma^c_{ea}+C^\beta_{\alpha a}C^c_{\beta b}.
		\label{CG3}
	\end{equation}
	Substituting (\ref{CG3}) in (\ref{commLG3}), we obtain (\ref{condLs2}).

Thus, relations (\ref{sysBa}) and (\ref{sysB}) are satisfied. Then the Dirac operator $\mathcal {D}_M^0$ can be obtained as the projection of the operator $B^a\eta_a+B$, where $B^a$ and $B$ are determined by (\ref{coefBB}), onto $M$ and we come to the projection map (\ref{dirac02}).
\end{proof}

From this theorem  we immediately obtain
\begin{slt}
\label{coroll1}
The generators
	\begin{equation}\nonumber
		 X' = \xi_X^a(x)\partial_{x^a} + \xi^\alpha_X(x,e_H)\Lambda^s_\alpha,\quad X\in\mathfrak{g},
	\end{equation}
of a representation of the Lie algebra $\mathfrak{g}$  in the space $V_\Psi$ are the symmetry operators of the Dirac operator $\mathcal{D}^0_M$ on the homogeneous space $M$.
\end{slt}

\section{The Dirac equation with an external gauge field admitting the motion group of a homogeneous space as a symmetry group}
\label{dirac-eq-gauge-field}

Here we consider the Dirac equation with an external gauge Yang--Mills field on a homogeneous space  $M$. Our aim is to find the Yang--Mills potentials for which the Dirac equation admits the motion group of the homogeneous space as a symmetry group. Let $V_K$  be a set of vector fields on $M$ transforming according to the fundamental representation of an $N$-dimensional gauge Lie group $K$.

A multiplet of $N$ spinor fields on  $M$ can be considered as a space $C^{\infty}(M,V)$ of functions on $M$ which take values on a linear space $V=V_K \otimes V_\Psi $.

The potential  $A_i,\ i=1,\dots,\dim M = 4$, of the gauge Yang--Mills field takes values  in the Lie algebra $\mathfrak{k}$ of the gauge group  $K$, $A_i = g A_i^{\bar{a}}(x)T_{\bar{a}}$, where  $g$  is the coupling constant and the generators $T_{\bar{a}}$ of the gauge group  take values in $V_K$:
\begin{equation}\nonumber
	[T_{\bar{a}},T_{\bar{b}}] = f^{\bar{c}}_{\bar{a}\bar{b}} T_{\bar{c}},\quad
	\bar{a},\bar{b},\bar{c} = 1,\dots,N.
\end{equation}
Here $f^{\bar{c}}_{\bar{a}\bar{b}}$ are the structure constants of the gauge group  $K$.

As a result, the Dirac equation on $M$ with an external non-Abelian gauge field can be written as
\begin{gather}\nonumber
	\mathcal{D}_M \psi = m \psi,\quad
	\mathcal{D}_M = \mathcal{D}_M^0 \mathbb{E}_N + i\hat\gamma^a\eta^j_a(x,e_H)A_j,
\end{gather}
where $\mathbb{E}_N$ is an identity matrix on the space $V_K$.

Let the Dirac equation possess the group of motions $G$ of the homogeneous space $M$ as a symmetry group. Then, the inclusion of an external non-Abelian gauge potential in the Dirac operator breaks its the symmetry, in the general case, since the Dirac operator with the external field, in contrast to $\mathcal {D}^0_M$, is no longer an invariant operator:
\begin{equation}
	\mathcal{D}_M = i \mathbb{E}_4\hat\gamma^a\eta^j_a(x,e_H)\partial_{x^j} +
	i\hat\gamma^a\left(\eta^\alpha_a(x,e_H)\Lambda^s_\alpha\mathbb{E}_4 + \eta^j_a(x,e_H)A_j  \right) + i\Gamma \mathbb{E}_4.
	\label{DM}
\end{equation}
We  will seek the Yang-Mills potentials $A_i$  for which  the symmetry group of the Dirac equation is the group of motions $G$. In this case, the operator $\mathcal{D}_M$ must be invariant under the group of motions and can be presented in the form of  (\ref{R1form}). Comparing (\ref{DM}) with (\ref{R1form}), we obtain
\begin{equation}
	B^a = i\hat\gamma^a\mathbb{E}_4,\quad
	B = i\Gamma\mathbb{E}_4,\quad
	\eta^\alpha_a(x,e_H)\Lambda_\alpha = \eta^\alpha_a\Lambda^s_\alpha\mathbb{E}_4 + \eta^j_a(x,e_H)A_j,	
	\label{kalB}
\end{equation}
and, hence,
\begin{equation}
	\eta^j_a(x,e_H)A_j = \eta^\alpha_a(x,e_H)\Lambda^k_\alpha,
	\label{cond0A}
\end{equation}
where $\Lambda^k_\alpha$ take values on the space  $V_K$.
Let us now multiply (\ref{cond0A}) by $\sigma^a_j(x, e_H)$ and perform summation over $a$. Here $\sigma^a_j(x, e_H)$ is the inverse matrix to $\eta^j_a(x, e_H)$.
Considering that
\begin{gather}\nonumber
	\sigma^a_j(x,e_H)\eta^\alpha_a(x,e_H) =
	\sigma^A_j(x,e_H)\eta^\alpha_A(x,e_H) - \sigma^\beta_j(x,e_H)\eta^\alpha_\beta(e_H) = \\ \nonumber
	\delta^\alpha_j - ( - \delta^\alpha_\beta) \sigma^\beta_j(x,e_H) = \sigma^\alpha_j(x,e_H),
	\label{}
\end{gather}
we finally obtain
\begin{equation}
	A_j = \sigma^\alpha_j(x,e_H)\Lambda^k_\alpha.
	\label{cond1A}
\end{equation}

From (\ref{cond1A}) it follows that the gauge group $K$ of the potentials $A_j$ is isomorphic to the isotropy subgroup $H$ of the homogeneous space $M$.

Also, when condition (\ref{cond1A}) is fulfilled,  the generators $\Lambda_\alpha$ of the representation on the space  $V_K\otimes V_\Psi$  can be written as:
\begin{equation}
	\Lambda_\alpha = \Lambda^s_\alpha\mathbb{E}_4 + \Lambda^k_\alpha.
	\label{factL}
\end{equation}
Substituting (\ref{factL}) in the commutation relations for $\Lambda_\alpha$, we see that
\begin{gather}\nonumber
	[\Lambda_\alpha,\Lambda_\beta] - C^\gamma_{\alpha\beta}\Lambda_\gamma =
	\left([\Lambda^s_\alpha,\Lambda^s_\beta] - C^\gamma_{\alpha\beta}\Lambda^s_\gamma\right)\mathbb{E}_4 +
	[\Lambda^k_\alpha,\Lambda^k_\beta] - C^\gamma_{\alpha\beta}\Lambda^k_\gamma = 0.
\end{gather}
The expression in parentheses vanishes according to Theorem  \ref{prop1}.
Consequently, the generators  $\Lambda^k_\alpha$ satisfy the commutation relations $[\Lambda^k_\alpha, \Lambda^k_\beta] = C^\gamma_{\alpha\beta} \Lambda^k_\gamma$. Thus, $\Lambda^k_\alpha$  are the generators of the isotropy subgroup $H$ on the representation space $V_K$ of the gauge group $K$. Then we have:
\begin{equation}\nonumber
	U(h) = \exp(h^\alpha\Lambda_\alpha^s\mathbb{E}_4)\exp(h^\beta\Lambda_\alpha^k) = U^s(h) U^k(h).
\end{equation}
For gauge fields of the form (\ref{cond1A}), the Dirac operator $\mathcal{D}_M$ takes the form of (\ref{R1form}):
\begin{equation}\nonumber
	\mathcal{D}_M = i\hat\gamma^a\left(\mathbb{E}_4 \eta^i_a(x,e_H)\partial_{i}+\eta^\alpha_a(x,e_H)\Lambda_\alpha\right)+i\Gamma\mathbb{E}_4.
	\label{}
\end{equation}
Relations (\ref{sysBa})--(\ref{sysB}) for the coefficients  $B^a$ and $B$ are the necessary and sufficient conditions for the operators (\ref{R1form}) to be invariant. Substituting (\ref{kalB}) in (\ref{sysBa})--(\ref{sysB}) and taking into account relations (\ref{condLs2}), we obtain the condition:
\begin{equation}
	\hat\gamma^a C^\beta_{a\alpha}\Lambda^k_\beta = 0,\quad
	\alpha,\beta = 1,\dots,\dim \mathfrak{h}.
	\label{cond2A}
\end{equation}
By virtue of the linear independence of the gamma matrices $\hat\gamma^a $, this condition is satisfied if and only if the structure constants $C^\beta_{a\alpha}$ are zero. In other words,  condition (\ref{cond2A}) is equivalent to \textit{reductivity} of the homogeneous space:
\begin{equation}
	[\mathfrak{m},\mathfrak{h}]\subset\mathfrak{m}.
	\label{cond2Aalg}
\end{equation}
Let us say that a Dirac operator $\mathcal{D}_M $ on homogeneous space $M$ with an invariant metric $g_M$ admits an external gauge field
if the motion group of  $M$ is the symmetry group of the Dirac equation.
Thus, we have
\begin{utv}
\label{prop4}
The Dirac operator $\mathcal{D}_M $ admits an external gauge field
if and only if the  homogeneous space is reductive.
\end{utv}
In this case, the isotropy subgroup $H$ is the gauge group. The potential of the gauge field is determined by the representation generators  $\Lambda^k_\alpha$  of the subgroup $H$ on $V_K$ (see (\ref{cond1A})).
Under the conditions of Theorem  \ref{prop4}, the  Dirac operator $\mathcal{D}_M$ can be redefined as the projection map of the operator
\begin{equation}
	\mathcal{D}_G = i\hat\gamma^a\left( \eta_a(g) + \Gamma \right)\in L(\hat{\mathcal{F}}).
	\label{Dg2}
\end{equation}
Comparing (\ref{Dg2}) and (\ref{dirac02}), we conclude that the external gauge potential (\ref{cond1A}) in the Dirac operator does not change the operator $\mathcal{D}_G$ and leads to the replacement
$$
	U(h) = U^s(h) \rightarrow U(h) = U^k(h)U^s(h).
$$
We are now in a position to give the following theorem:
\begin{utv}
\label{prop5}
Let the Dirac equation on a homogeneous space $M$ with an invariant metric $g_M$ and an external gauge field  admits a group of motions $G$ as the symmetry group. Then the Dirac equation is equivalent to the system of equations
\begin{equation}
	\mathcal{D}_G \psi(g) = m\psi(g),\quad (\eta_\alpha + \Lambda^s_{\alpha}\mathbb{E}_N + \Lambda^k_{\alpha})\psi(g) = 0,
	\label{sysGK}
\end{equation}
where $\mathcal{D}_G$ is the Dirac operator without an external gauge field on the Lie group $G$.
\end{utv}

\section{The Dirac equation in an $\mathbb{R}^2\otimes \mathbb{S}^2$ space}
\label{dirac-eq-r2-s2}

The homogeneous space $M$ with the transformation group $G=\mathbb{R}^2\otimes SO(3)$ and the isotropy subgroup $SO(2)$ is topologically isomorphic to the Cartesian product of a two-dimensional plane $\mathbb{R}^2$ and a sphere $\mathbb{S}^2$. Let us introduce a local coordinate system  $(t,x,\phi,\theta)$  on $M$, where $(t,x)\in\mathbb {R}^2$, $\phi\in (0,2\pi ) $, and $\theta\in (0, \pi)$. The invariant metric in the local coordinates is given as
\begin{equation}
	ds^2 = dt^2 - c^{-1}_1 dx^2 - c^{-1}_2 d\Omega_2,\quad
	d\Omega_2 = d\theta^2 + \sin\theta^2d\phi^2,\quad c_1 > 0, c_2 > 0.
	\label{ds2_so3}
\end{equation}
Denote by $\{e_1, e_2, e_3, e_4, e_5 \}$  a fixed basis of the Lie algebra $\mathfrak{g}=\mathbb {R}^ 2 \times \mathfrak {so}(3)$ of the Lie group $G$, where $\mathbb{R}^2=\{e_1, e_2 \},\ \mathfrak {so}(3)=\{e_3, e_4, e_5 \}$. The non-zero commutation relations are
\begin{gather}\nonumber
	[e_3,e_4] = e_5,\quad
	[e_5,e_3] = e_4,\quad
	[e_4,e_5] = e_3.
\end{gather}
The isotropy subalgebra $\mathfrak {h}$ is generated by the basis element $e_5$: $H =\exp(h e_5) $, $h\in (0,2 \pi) $. The invariant metric (\ref{ds2_so3}) is defined by a non-degenerate  2-form $G^{ab} = \operatorname {diag} (1,-c_1,-c_2,-c_2) $ on the basis $\{e_a \}$, where $c_1 > 0, c_2 > 0$.

The local coordinates $(t, x, \phi, \theta, h) $  are defined by canonical coordinates of the second kind on the Lie group $G$:
\begin{equation}
	g(t,x,\phi,\theta,h) = e^{h e_5} e^{(\theta-\frac{\pi}{2}) e_4} e^{\phi e_3} e^{x e_2} e^{t e_1},
	\quad (t,x,\phi,\theta)\in M,\quad h\in H.
	\label{mycanG}
\end{equation}
In the coordinates (\ref{mycanG}), the  right-invariant basis of the 1-forms  $\sigma^A (x, h) $ reads
\begin{gather}
	\nonumber
	\sigma^1 = -dt, \quad
	\sigma^2 = -dx, \quad
	\sigma^3 = -\sin\theta\cos{h} d\phi + \sin{h}d\theta, \\ \nonumber
	\sigma^4 = -\sin\theta\sin{h} d\phi - \cos{h}d\theta, \quad
	\sigma^5 = -\cos\theta d\phi - d h.
\end{gather}
Since the isotropy subgroup is one-dimensional, the gauge group that does not change the symmetry of the equation is Abelian. For $K=U(1)$, we take the potential (\ref{cond1A}) of an external electromagnetic field:
\begin{equation}
	A_1 = A_2 = A_4 = 0,\quad A_3 = -\varepsilon\cos{\theta},\quad \varepsilon\in\mathbb{R}.
	\label{myAA}
\end{equation}
The potential (\ref{myAA}) describes a stationary magnetic field depending on the variable $\theta $ with the strength tensor
\begin{equation}\nonumber
	F = F_{ij}dx^i\wedge dx^j = -\varepsilon\sin{\theta} d\phi \wedge d\theta.
	\label{}
\end{equation}
From Theorem  \ref{prop5} it follows that the Dirac equation on the homogeneous space $M$ with the magnetic field (\ref{myAA}) is equivalent to the system of equations
\begin{equation}
	(i\hat\gamma^a\eta_a(g) - m)\psi(g) = 0,\quad
	(-\partial_{h} + \frac{1}{2 c_2}\hat\gamma^3\hat\gamma^4+\varepsilon)\psi(g)=0,
	\label{sysSO3dirac}
\end{equation}
where $\eta_a(g)$ are given by
	\begin{gather}\nonumber
		\eta_1 = -\partial_t,\quad
		\eta_2 = -\partial_x,\quad
		\eta_5 = -\partial_h,\\ \nonumber
		\eta_3 = -\frac{\cos{h}}{\sin{\theta}}\partial_{\phi}+
		\sin{h}\partial_{\theta}+\cos{h}\cot{\theta}\partial_h,\quad
		\eta_4 = -\frac{\sin{h}}{\sin{\theta}}\partial_{\phi}-
		\cos{h}\partial_{\theta}+\sin{h}\cot{\theta}\partial_h.
	\end{gather}
The left invariant vector fields
 \begin{gather}\nonumber
	\xi_1 = \partial_t,\quad
	\xi_2 = \partial_x,\quad
	\xi_3 = \partial_\phi,\\ \nonumber
	\xi_4 = -\cot{\theta}\sin{\phi}\partial_\phi + \cos{\phi}\partial_\theta +
	\frac{\sin{\phi}}{\sin{\theta}}\partial_h,\quad
	\xi_5 = -\cot{\theta}\cos{\phi}\partial_\phi - \sin{\phi}\partial_\theta +
	\frac{\cos{\phi}}{\sin{\theta}}\partial_h
\end{gather}
on the Lie group $G$ provide  symmetry operators for the Dirac equation (\ref{sysSO3dirac}).

Let us find the spectrum of the Dirac equation (\ref{sysSO3dirac}). The operators $\xi_1$ and $\xi_2$ form an Abelian algebra $\mathbb {R}^2$ and allow one to separate the variables $t$ and $x$:
\begin{equation}\label{pf2}
	\psi(g) = e^{-i (\omega t + j_1 x)}f(\phi,\theta,h),\quad
	i\xi_1\psi(g) = \omega \psi(g),\quad
	i\xi_2\psi(g) = j_1 \psi(g).
\end{equation}
The operators $\xi_3$, $\xi_4$, and $\xi_5$ form a Lie algebra $\mathfrak{so}(3)$. We take into account this symmetry algebra using the noncommutative integration method \cite{SpSh1}.

According to this method, we will solve equations (\ref{sysSO3dirac}) together with the system
\begin{equation}
	\left(\xi_A(g) + l_A(q,\lambda)\right) \psi(g) = 0,\quad
	\left(\eta_A(g) + \overline{l_A(q',\lambda)}\right)\psi(g) = 0,\quad
	(A = 3,4,5).
	\label{sysXL}
\end{equation}
We call  $l_A(q, \lambda)$ \textit{the  operators of $\lambda $-representation} for the Lie algebra $\mathfrak {so}(3)$ on the space $L_2(Q, d \mu (q))$ of functions determined on the Lagrangian submanifold $Q$ of a coadjoint orbit (K-orbit) of the Lie algebra $\mathfrak{so}(3)$ \cite{ShSO3}. The submanifold $Q$ has the topology of a cylinder: $ q =\alpha + i\beta \in Q,\ \alpha\in(0, 2 \pi),\ \beta\in\mathbb {R}^1$.  The covector $\lambda= (j_2, 0,0) \in \mathfrak {so}^{*}(3), \ j_2 > 0$, parameterizes a non-degenerate  K-orbit.  The operators $-il_A(q, \lambda)$ are hermitian in $ L_2 (Q, d \mu(q))$ with respect to the inner product
\begin{equation}\nonumber
	(\psi_1(q),\psi_2(q)) = \int_Q \overline{\psi_1(q)}\psi_2(q)d\mu(q),\quad
	d\mu(q) = \frac{(2 j_2 + 1)!}{2^{j_2}(j_2 !)^2}\frac{dq\wedge d\overline{q}}{(1+\cos(q-\overline{q}))^{j_2 + 1}}
\end{equation}
and are given by the following equations \cite{ShSO3}:
\begin{gather}\nonumber
	l_3(q,\lambda) = -i(\sin(q)\partial_q - j_2 \cos(q)),\quad
	l_4(q,\lambda) = -i(\cos(q)\partial_q + j_2 \sin(q)),\quad
	l_5(q,\lambda) = \partial_q.
	\label{lprSO3}
\end{gather}
In the framework of the non-commutative integration method, the system of equations (\ref{sysXL}) is represented as a generalized Fourier transform \cite{ShBr11, ShBr14}:
\begin{equation}\nonumber
	f(\phi,\theta,h) = \sum_{j_2 = 0}^{\infty}(2 j_2 + 1)\int_Q \psi(q,q',j_2) D^{j_2}_{q \overline{q'}}(\phi,\theta,h) d\mu(q')d\mu(q),
	\label{}
\end{equation}
where $D^{j_2}_{q\overline{q}'}(\phi,\theta,h)$ is given by the formula
\begin{gather}\nonumber
	D^{j_2}_{q \overline{q}'}(\phi,\theta,h) =  \frac{2^{j_2}(j_2 !)^2}{(2 j_2)!}
	\bigg{(}\sin{\theta} \cos{\phi} +
	\cos\left(h+ \overline{q}' \right)(\cos{q}\cos{\phi}-i\sin{\phi})-i\sin{\theta}\cos{q}\sin{\phi}- \\ \nonumber
	i\cos{\theta}\sin{q}+\sin\left(h+\overline{q}'\right)(-i\cos{\theta} \cos{\phi}-\cos{\theta} \cos{q}\sin{\phi}+\sin{\theta} \sin{q})\bigg{)}^{j2}.
	\label{}
\end{gather}
Substituting (\ref{pf2}) in the Dirac equation (\ref{sysSO3dirac})  and taking into account relation (\ref{sysXL}), we obtain a system of equations for $\psi(q,q',j_2)$:
\begin{gather}\label{sysQQ1}
	\left(- \omega \hat\gamma^1 - j_1 \hat\gamma^2 + i \hat\gamma^3 l_3(q',j_2) + i \hat\gamma^4 l_4(q',j_2) - m\right)\psi(q,q',j_2) = 0,\\
	\left( l_5(q',j_2) +\frac{1}{2 c_2}\hat\gamma^3\hat\gamma^4 + \varepsilon\right) \psi(q,q',j_2) = 0	\label{sysQQ2}
\end{gather}
From (\ref{sysQQ2}) it follows that
\begin{equation}
	\psi(q,q',j_2) = \exp\left( -\left[\frac{1}{2 c_2}\hat\gamma^3\hat\gamma^4 + \varepsilon\right] q' \right) R(q,\omega,j_1,j_2).		
	\label{forR}
\end{equation}
Substituting (\ref{forR}) in the reduced equation (\ref{sysQQ1}), we find the linear system of equations 
\begin{equation}\label{sysR}
	\left( \omega \hat\gamma^1 + j_1 \hat\gamma^2 + \left(j_2 + \frac{1}{2}\right)\hat\gamma^3 + \varepsilon \hat\gamma^4 + m\right)R(q,\omega,j_1,j_2) = 0.	
\end{equation}
Finally, from the compatibility conditions for this system, we obtain the following relation determining the spectrum of the Dirac equation:
\begin{equation}\nonumber
	\omega^2 =  j_1^2 c_1 + \left[\varepsilon^2 + \frac{1}{4} + j_2 (j_2 + 1)\right] c_2 + m^2.
\end{equation}
Note that the spectral parameter $j_1$ is continuous because the variable $x$ takes values on the non-compact space $\mathbb{R}^1$. The parameter $j_2 $ has the physical meaning of an orbital quantum number and it takes only a discrete set of values. The expression for $j_2(j_2+1)$ corresponds to the spectrum of the Laplace operator on the sphere $\mathbb{S}^2$.

In the context of Kirillov's coadjoint orbit method \cite{Kirr}, the quantization of $j_2$ is prescribed  by the condition of \textit{integer-valued orbits}. Note that this condition is equivalent to the Kostant-Souriau condition in the context of the geometric quantization method \cite{Konst}.

The basis of solutions for the linear system of equations (\ref{sysR}) reads
\begin{gather}\nonumber
	R^{s=1}(\omega,j_1,j_2) = \left(
		\begin{array}[c]{c}
		\frac{i\sqrt{c_2}(j_2+1/2)-\sqrt{c_1}j_1}{\omega + m}\\
		\frac{\varepsilon \sqrt{c_2}}{\omega + m}\\
		0\\
		1
		\end{array}
	\right),\\ \nonumber
	R^{s=-1}(\omega,j_1,j_2) = \left(
		\begin{array}[c]{c}
		-\frac{\varepsilon \sqrt{c_2}}{\omega + m}\\	
		-\frac{i\sqrt{c_2}(j_2+1/2)+\sqrt{c_1}j_1}{\omega + m}\\
		1\\
		0
		\end{array}
		\right),	
\end{gather}
where the spin index $s$ enumerates the basis vectors. The basis of the original Dirac equation (\ref{sysSO3dirac}) can be written in the framework of the noncommutative integration method as:
\begin{gather}\nonumber
	\psi_{\tau}(t,x,\phi,\theta) = e^{-i(\omega t + j_1 t)}D^{j_2}_q(\phi,\theta) R_\tau,\quad \tau = (j_0,j_1,j_2,q,s),\\ \nonumber
	D^{j_2}_q(\phi,\theta) = \int_Q e^{-(\frac{1}{2 c_2}\hat\gamma^3\hat\gamma^4 + \varepsilon) q'}D^\lambda_{q\overline{q'}}(\phi,\theta,e_H)d\mu(q').
	\label{}
\end{gather}
Here $\tau$  denotes the set of quantum numbers for the basis of solutions.

\section{The Dirac equation in de-Sitter space}
\label{de-sitter}

Consider de Sitter space $M$ as a homogeneous space with de Sitter isometry group of transformations $G=SO(1,4)$ and an isotropy Lorentz subgroup $H=SO(1,3)$. The space  $M$ is topologically isomorphic to $R^1 \times S^3$ and has constant positive curvature.

The de Sitter group $SO(1,4)$ is a rotation group of the 5-dimensional pseudo-Euclidean space with the metric $G_{AB}=\operatorname {diag} (1, -1, -1, -1, -1)$.
The algebra $\mathfrak{g} =\mathfrak {so}(1,4)$ of the de Sitter group can be defined in terms of the basis $\{E_ {AB}\mid A<B \} $
by the following commutation relations:
\begin{equation}\nonumber
	[E_{AB},E_{CD}] = G_{AD}E_{BC} - G_{AC}E_{BD} + G_{BC}E_{AD}-G_{BD}E_{AC},
	\label{}
\end{equation}
where $A,B,C,D = 1,\dots,5$. The basis $E_{AB}$ can be written as
\begin{equation}\nonumber
	E_{ab} = e_{ab} ,\ (a < b),\quad E_{a 5} = e_a/\varepsilon,\quad
	[e_a,e_b] = \varepsilon^2 e_{ab},\quad
	a,b = 1,\dots,4.
\end{equation}
Here the basis $e_{ab}$ forms an isotropy subalgebra $\mathfrak {h}=\mathfrak {so}(1,3)$, and $\varepsilon $ is a parameter defining the curvature of  de Sitter space, $R=12\varepsilon^ 2$. Define  canonical coordinates of the second kind for  a Lie group $G$ by the formula
\begin{equation}
	g(t,x,y,z,h_1,h_2,h_3,h_4,h_5,h_6) =
	e^{h_6 e_{34}} e^{h_5 e_{24}} e^{h_4 e_{23}} e^{h_3 e_{14}}
	e^{h_2 e_{13}} e^{h_1 e_{12}}
	e^{z e_4} e^{y e_3} e^{x e_2} e^{t e_1}.
	\label{can_so14}
\end{equation}
Here $(t, x, y, z)$ are local coordinates on the de Sitter space $M$ and $h_{ab}$ are local coordinates on  the isotropy subgroup $H$. The right-invariant 1-forms $\sigma^{\alpha}(x, e_H)$ in the canonical coordinates (\ref{can_so14}) are written as
\begin{gather}\nonumber
	\sigma^{12} = \varepsilon \sin{\varepsilon x} dt -dh_1,\quad
	\sigma^{13} = \varepsilon \cos{\varepsilon x}\sin{\varepsilon y} dt -dh_2,\\ \nonumber
	\sigma^{14} = \varepsilon \cos{\varepsilon x}\cos{\varepsilon y}\sin{\varepsilon z} dt - dh_3,\quad
	\sigma^{23} = \varepsilon \sin{\varepsilon y} dx - dh_4,\\ \nonumber
	\sigma^{24} = \varepsilon \cos{\varepsilon y}\sin{\varepsilon z} dx - dh_5,\quad
	\sigma^{34} = \varepsilon \sin{\varepsilon z} dy - dh_6.
\end{gather}
For the isotropy subgroup $H$, we choose generators of a representation on a space $V$, $\dim V=4$, in the form
\begin{equation}\nonumber
	\Lambda_{ab} = \frac{1}{4}[\gamma^a,\gamma^b],\quad
    \{ \gamma^a,\gamma^b \} = 2 \eta^{ab}\mathbb{E}_4,\quad \eta^{ab} = \operatorname{diag}(1,-1,-1,-1),\quad a < b.
\end{equation}
The potentials of the external gauge field for which the de Sitter group is the symmetry group of the Dirac equation are given by
\begin{gather}\nonumber
	A_1 = \varepsilon \sin{\varepsilon x}\Lambda_{12}+\varepsilon\cos{\varepsilon x}\sin{\varepsilon y}\Lambda_{13} + \varepsilon \cos{\varepsilon x}\cos{\varepsilon y}\sin{\varepsilon z}\Lambda_{14},\\ \nonumber
	A_2 = \varepsilon \sin{\varepsilon y}\Lambda_{23}+\varepsilon\cos{\varepsilon y}\sin{\varepsilon z}\Lambda_{24},\\ \nonumber
	A_3 = \varepsilon \sin{\varepsilon z}\Lambda_{34},\quad
	A_4 = 0.
\end{gather}
These potentials define the gauge field strength tensor
\begin{gather}\label{desitterF}
	F_{12} = -\varepsilon^2\cos{\varepsilon x}\cos^2{\varepsilon y}\cos^2{\varepsilon z}\Lambda_{12},\quad
	F_{13} = -\varepsilon^2\cos{\varepsilon x}\cos{\varepsilon y}\cos^2{\varepsilon z}\Lambda_{13},\\ \nonumber
	F_{14} = -\varepsilon^2\cos{\varepsilon x}\cos{\varepsilon y}\cos{\varepsilon z}\Lambda_{14},\quad
	F_{23} = -\varepsilon^2\cos{\varepsilon y}\cos^2{\varepsilon z}\Lambda_{23},\\ \nonumber
	F_{24} = -\varepsilon^2\cos{\varepsilon y}\cos{\varepsilon z}\Lambda_{24},\quad
	F_{34} =-\varepsilon^2\cos{\varepsilon z}\Lambda_{34}.
\end{gather}
The solutions of the Dirac equation with the external gauge field (\ref{desitterF}) can be obtained using the noncommutative integration method \cite{SpSh1}. The structure of the K-orbits for the  de Sitter group is rather complicated \cite{Kirr}. Detail consideration of the noncommutative integration method and  exact solution construction for the Dirac equation in the case under consideration is beyond the scope of the present work and will be the subject of a separate study.

\section{Conclusion remarks}
\label{conclusion}

We have obtained the Yang--Mills potentials for which  the symmetry group of the Dirac equation is a group of transformations of a homogeneous space $M$. Such gauge fields are found by  representation of the Dirac equation in terms of an invariant matrix differential operator of the first order. The matrix coefficients of this operator satisfy an algebraic system of equations that is equivalent to the condition that the  gauge fields do not change the symmetry of the Dirac equation, and this is possible if  the homogeneous space $M$ is reductive. For the reductive space $M$, these gauge fields are determined by the Lie symmetry algebra of the equation. The gauge group in this case is isomorphic to a subgroup $H$ of the isotropic homogeneous space $M$.

The noncommutative integration method \cite{SpSh1} can be effectively used to construct solutions of the Dirac  equation with the gauge fields (\ref{cond1A}).

Following this method, we  have found the spectrum and the basis of solutions of the Dirac equation in an  $\mathbb{R}^2 \times \mathbb {S}^2$ space with an invariant metric. This is a reductive homogeneous space, and the potentials (\ref{cond1A}) describe the external magnetic field.
The Dirac equation spectrum is shown to be represented as a continuous part corresponding to a non-compact subgroup of $\mathbb {R}$ and a part corresponding to the spectrum of the Laplace operator on a two-dimensional sphere $\mathbb {S}^2$. Also, we have obtained explicit formulae for the Yang--Mills potentials conserving the symmetry group of the Dirac equation in de Sitter space.

Exact solutions of the Dirac equation with the  gauge fields of the form (\ref{cond1A}) can be used to investigate the effects of vacuum polarization and particle creation in homogeneous cosmological models \cite{Birrel, ShBr14}. Note that seaking exact solutions to the Dirac equation on homogeneous spaces is a mathematically attractive problem, but it is far from trivial even when  external fields are absent \cite{Birrel}. On the other hand, the results obtained contribute to the study of mathematical properties of external fields with internal symmetry.

\section*{Acknowledgments}

The work was partially supported by Tomsk State University Competitiveness Improvement Program and by program ‘Nauka’ under contract No. 1.676.2014/ K.

\end{document}